\documentclass[submission,copyright,creativecommons]{eptcs}
\providecommand{\event}{PxTP 2015} 

\usepackage{amsthm}
\usepackage{amsmath}
\usepackage{xcolor}
\usepackage{graphicx}
\usepackage{multicol}

\newtheorem{definition}{Definition}
\newtheorem{theorem}{Theorem}
\newtheorem{lemma}{Lemma}

\newcommand{\Sh}[2]{\mathsf{Sh}(#1, #2)}
\newcommand{\Dp}[2]{\mathsf{Dp}(#1, #2)}
\newcommand{\Et}[3]{\mathsf{ET}(#1, #2, #3)}
\newcommand{\Ets}[2]{\mathsf{ET}(#1, #2)}

\title{Importing SMT and Connection proofs as expansion trees}
\author{Giselle Reis
\institute{INRIA-Saclay, France}
\email{giselle.reis@inria.fr}
}
\def\titlerunning{Importing SMT and Connection proofs as expansion trees}
\def\authorrunning{Giselle Reis}

\begin{document}
\maketitle

\begin{abstract}
Different automated theorem provers reason in various deductive systems and, thus,
produce proof objects which are in general not compatible. To understand and
analyze these objects, one needs to study the corresponding proof theory,
and then study the language used to represent proofs, on a prover by prover basis. 
In this work we present an implementation that takes SMT and Connection proof
objects from two different provers and imports them both as expansion trees. By
representing the proofs in the same framework, all the algorithms and tools
available for expansion trees (compression, visualization, sequent calculus
proof construction, proof checking, etc.) can be employed uniformly. The
expansion proofs can also be used as a validation tool for the proof objects
produced.
\end{abstract}
\vspace{-0.5cm}

\section{Introduction}
\vspace{-0.2cm}

The field of proof theory has evolved in such a way to create the most various
proof abstractions. Natural deduction, sequent calculus, resolution, tableaux,
SAT, are only a few of them, and even within the same formalism there might be
many variations. As a result, automated theorem provers will generate different
proof objects, usually corresponding to their internal proof representation. The
use of distinct formats has some disadvantages: provers cannot recognize each
others proofs; proofs cannot be easily compared; all analysis and algorithms need
to be developed on a prover by prover basis.

GAPT is a framework for proof theory that is able to represent, process and
visualize proofs. Currently it implements the sequent calculus LK
(with or
without equality rules) for first and higher order classical logic, Robinson's
resolution calculus \cite{robinson65}, the schematic calculus LKS \cite{weller13}
and expansion trees \cite{miller87}. GAPT also provides algorithms for
translating proofs between some of these formats, for cut-elimination (reductive
methods \`{a} la Gentzen \cite{gentzen35} and CERES \cite{ceres}), and
for cut-introduction (proof compression) \cite{hetzl14}, as well as an interactive proof
visualization tool \cite{dunchev13}. But all these tools depend on having proofs
to operate on.

In this work we show how to parse and translate SMT and Connection proofs from
veriT and leanCoP, respectively, into expansion proofs in GAPT. 
SMT are unsatisfiability proofs with respect to some theory and, in veriT, these
are represented by resolution refutations of a set including (instances of) the
axioms of the theory considered and the negation of the input formula.
Connection proofs decide first-order logic formulas by connecting literals of
opposite polarity in the clausal normal form of the input. These different
conceptions of proofs will be unified under the form of expansion proofs, which
can be considered a compact representation of sequent calculus proofs.

The advantages of this work is three-fold. 
First of all, the use of expansion proofs provides a compact
representation for otherwise big and hard to grasp proof objects. Using this
representation and GAPT's visualization tool, it is easy to see the theorem
that was proved and the instances of quantified formulas used. 
Second of all, the use of a common representation facilitates the comparison of
proofs and makes it possible to run and analyse algorithms developed for this
representation without the need to adapt it to different formats. In particular,
we have been using the imported proofs for experimenting proof compression via
introduction of cuts \cite{hetzl14}.
Finally, it provides a simple sanity-check procedure and the
possibility of building LK proofs.

This paper is organized as follows. Section \ref{sec:exp_proofs} defines basic
concepts and extends the usual definition of expansion trees to accommodate
polarities. 
Section \ref{sec:importing} explains how to extract the necessary information
from both formats and how it is then used to build expansion trees.
Section \ref{sec:results} presents the
results of the transformation applied to a database of proofs in the considered
formats. It also discusses the advantages of having the proofs as expansion
trees. Section \ref{sec:related} discusses some related work and, finally,
Section \ref{sec:conclusion} concludes the paper pointing to
future work.
\vspace{-0.5cm}

\section{Expansion proofs}
\label{sec:exp_proofs}
\vspace{-0.2cm}

We will work in the setting of first-order classical logic. We introduce now a
few basic concepts.

\begin{definition}[Polarity in a sequent]
\label{def:polarity_in_seq}
Let $S = A_1, ..., A_n \vdash B_1, ..., B_m$ be a sequent. We will say that
formulas on the left side of $\vdash$, i.e, $A_1, ..., A_n$ have \emph{negative}
polarity while formulas on the right, i.e., $B_1, ..., B_m$ have \emph{positive}
polarity.
\end{definition}

\begin{definition}[Polarity]
Let $F$ be a formula and $F'$ a sub-formula of $F$. Then we can define the
\emph{polarity} of $F'$ in $F$, i.e., $F'$ can be \emph{positive} or
\emph{negative} in $F$, according to the following criteria:
\begin{itemize}
  \item If $F \equiv F'$, then $F'$ has the same polarity as $F$.
  \item If $F \equiv A \wedge B$ or $F \equiv A \vee B$ or $F \equiv \forall x.
  A$ or $F \equiv \exists x. A$ and $F$ is \emph{positive} (\emph{negative}),
  than $A$ and $B$ are \emph{positive} (\emph{negative}).
  \item If $F \equiv A \rightarrow B$ and $F$ is \emph{positive}
  (\emph{negative}), then $A$ is \emph{negative} (\emph{positive}) and $B$ is
  \emph{positive} (\emph{negative}).
  \item If $F \equiv \neg A$ and $'$ is \emph{positive} (\emph{negative}) then
  $A$ is \emph{negative} (\emph{positive}).
\end{itemize}
%
%
\end{definition}

Throughout this document we will use $0$ for negative polarity, $1$ for positive
polarity and $\overline{p}$ to denote the opposite
polarity of $p$, for $p \in \{0,1\}$.

\begin{definition}[Strong and weak quantifiers]
\label{def:strong_weak_quant}
Let $F$ be a formula. If $\forall x$ occurs positively (negatively) in $F$, then
$\forall x$
is called a \emph{strong (weak) quantifier}. If $\exists x$ occurs positively
(negatively) in $F$, then $\exists x$ is called a \emph{weak (strong)
quantifier}. 
\end{definition}

\emph{Strong} quantifiers in a sequent will be those introduced by the
inferences $\forall_r$ and $\exists_l$ in a sequent calculus proof.

Expansion proofs are a compact representation for first and higher order sequent
calculus proofs. They can be seen as a generalization of Gentzen's
mid-sequent theorem to formulas which are not necessarily prenex
\cite{miller87}. Expansion proofs are composed by expansion trees. 
An expansion tree of a formula $F$ has this formula as its root. Leaves are atoms
occurring in $F$ and inner nodes are connectives or a quantified sub-formula of
$F$. The edges from quantified nodes to its children are labelled with terms
that were used to instantiate the outer-most quantifier.
We extend the
original definition with the notion of formula polarity and use $\Pi$ and
$\Lambda$ for strong and weak quantifiers respectively in expansion trees.

\begin{definition}[Expansion tree]
\emph{Expansion trees} and a function $\Sh{E}{p}$ (for \emph{shallow}), that
maps an expansion tree $E$ to a formula with polarity $p \in \{ 0, 1 \}$, 
are defined inductively as follows:
\begin{itemize}
  \item If $A$ is an atom, then $A$ is an expansion tree with top node $A$ and
  $\Sh{A}{p} = A$ for any choice of $p$.
  \item If $E_0$ is an expansion tree, then $E = \neg E_0$ is an expansion tree
  with $\Sh{E}{\overline{p}} = \neg \Sh{E_0}{p}$.
  \item If $E_1$ and $E_2$ are expansion trees and $\circ \in
  \{\wedge,\vee\}$, then $E = E_1 \circ E_2$ is an expansion tree with
  $\Sh{E}{p} = \Sh{E_1}{p} \circ \Sh{E_2}{p}$.
  \item If $E_1$ and $E_2$ are expansion trees, then $E = E_1 \rightarrow E_2$
  is an expansion tree with $\Sh{E}{p} = \Sh{E_1}{\overline{p}} \rightarrow
  \Sh{E_2}{p}$.
  \item If $\{t_1, ..., t_n\}$ is a set of terms and $E_1, ..., E_n$ are
  expansion trees with $\Sh{E_i}{p} = A[x/t_i]$, then $E = \Lambda x. A +^{t_1}
  E_1 ... +^{t_n} E_n$ (denoting a node with $n$ children) is an expansion tree
  with $\Sh{E}{0} = \forall x. A$ and $\Sh{E}{1} = \exists x. A$.
  \item If $E_0$ is an expansion tree with $\Sh{E_0}{p} = A[x/\alpha]$ for an
  Eigenvariable $\alpha$, then $E = \Pi x. A +^{\alpha} E_0$ is an expansion tree
  with $\Sh{E}{0} = \exists x. A$ and $\Sh{E}{1} = \forall x. A$.
\end{itemize}
\end{definition}

Expansion trees can be mapped to a quantifier free formula via the \emph{deep}
function, which we also redefine taking the polarities into account.

\begin{definition}
We define the function $\Dp{\cdot}{p}$ (for \emph{deep}), $p \in \{0, 1\}$, 
that maps an expansion tree to a quantifier free formula of polarity $p$ as:
\begin{multicols}{2}
\begin{itemize}
  \item $\Dp{A}{p} = A$ for an atom $A$.
  \item $\Dp{\neg A}{p} = \neg \Dp{A}{\overline{p}}$
  \item $\Dp{A \circ B}{p} = \Dp{A}{p} \circ \Dp{B}{p}$ for $\circ \in \{\wedge,
  \vee\}$
  \item $\Dp{A \rightarrow B}{p} = \Dp{A}{\overline{p}} \rightarrow \Dp{B}{p}$
  \item $\Dp{\Lambda x. A +^{t_1} E_1 ... +^{t_n} E_n}{0} = \bigwedge_{i=1}^n
  \Dp{E_i}{0}$
  \item $\Dp{\Lambda x. A +^{t_1} E_1 ... +^{t_n} E_n}{1} = \bigvee_{i=1}^n
  \Dp{E_i}{1}$
  \item $\Dp{\Pi x. A +^{\alpha} E}{p} = \Dp{E}{p}$
\end{itemize}
\end{multicols}
\end{definition}

\begin{definition}[Expansion sequent]
An \emph{expansion sequent} $\varepsilon$ is denoted by $E_1, ..., E_n \vdash
F_1, ..., F_m$ where $E_i$ and $F_i$ are expansion trees. Its \emph{deep
sequent} is the sequent $\Dp{E_1}{0}, ..., \Dp{E_n}{0} \vdash \Dp{F_1}{1}, ...,
\Dp{F_m}{1}$ and its \emph{shallow sequent} is $\Sh{E_1}{0}, ..., \Sh{E_n}{0}
\vdash \Sh{F_1}{1}, ..., \Sh{F_m}{1}$.
\end{definition}

An expansion sequent may or may not represent a proof. To decide whether this is
the case, we need to reason on the \emph{dependency relation} in the sequent.

\begin{definition}[Domination]
A term $t$ is said to \emph{dominate} a node $N$ in an expansion tree if it labels a
parent node of $N$.
\end{definition}

\begin{definition}[Dependency relation]
Let $\varepsilon$ be an expansion sequent and let $<_{\varepsilon}^0$ be the binary relation on the
occurrences of terms in $\varepsilon$ defined as: $t <_{\varepsilon}^0 s$ if there is an $x$ free in
$s$ that is an eigenvariable of a node dominated by $t$. Then $<_{\varepsilon}$, the
transitive closure of $<_{\varepsilon}^0$, is called the \emph{dependency relation} of $E$.
\end{definition}

\begin{definition}[Expansion proof]
An expansion sequent is considered an \emph{expansion proof} if its deep sequent
is a tautology and the dependency relation is acyclic. 
\end{definition}
%
Intuitively, the dependency relation gives an ordering of quantifier inferences
in a sequent calculus proof of the shallow sequent of $\varepsilon$. That is, $t
<_{\varepsilon} s$
means that the existential quantifiers instantiated with $t$ must occur lower in
the proof than those instantiated with $s$. Using this relation it is possible
to build an LK proof from an expansion proof \cite{miller87}.
\vspace{-0.5cm}

\section{Importing}
\label{sec:importing}
\vspace{-0.2cm}

GAPT\footnote{\url{https://github.com/gapt/gapt}} is a framework for proof
transformations implemented in the programming language Scala. It supports
different proof formats, such as LK (with or without equality) for first and
higher order logic, Robinson's resolution calculus \cite{robinson65}, the
schematic calculus LKS \cite{weller13} and, more recently, expansion trees. It
provides various algorithms for proofs, such as reductive cut-elimination
\cite{gentzen35},
cut-elimination by resolution \cite{ceres}, cut-introduction \cite{hetzl14},
Skolemization, and translations between the proof formats.
GAPT also comes with \texttt{prooftool} \cite{dunchev13}, an interactive proof
visualization tool supporting all these formats.

VeriT and leanCoP are automated theorem provers that produce unsatisfiability (in
the shape of a resolution refutation) and connection proofs respectively. Both
output the proof objects to a structured text file, having in
common the fact that all inferences are listed with the operands and the
conclusion. We have implemented parsers (using Scala's parser combinators) for
both formats in GAPT. By taking the necessary information of each proof file and
processing it accordingly, we can build expansion proofs. We explain
the kind of processing needed for each format in Sections \ref{sec:verit} and
\ref{sec:leancop}.

The expansion tree of a formula with associated substitutions to its bound
variables can be defined as follows:

\begin{definition}
\label{def:formula_to_et}
Let $F$ be a formula in which all bound variables have pairwise distinct names,
$\Sigma$ a set of substitutions
for these variables and $p \in \{ 0, 1\}$ a polarity. 
Assume that each strong quantifier in $F$ is bound to exactly one term in
$\Sigma$.  
We define the function $\Et{F}{\Sigma}{p}$ that translates a formula to an
expansion tree as follows:
\begin{itemize}
  \item $\Et{A}{\Sigma}{p} = A$, where $A$ is an atom.
  \item $\Et{\neg A}{\Sigma}{p} = \neg \Et{A}{\Sigma}{\overline{p}}$.
  \item $\Et{A \circ B}{\Sigma}{p} = \Et{A}{\Sigma}{p} \circ \Et{B}{\Sigma}{p}$,
  for $\circ \in \{\wedge,\vee\}$.
  \item $\Et{A \rightarrow B}{\Sigma}{p} = \Et{A}{\Sigma}{\overline{p}}
  \rightarrow \Et{B}{\Sigma}{p}$.
  \item $\Et{\forall x. A}{\Sigma}{0} = \Lambda x. A +^{t_1}
  \Et{A\sigma_1}{\{\sigma_1\}}{0} ... +^{t_n} \Et{A\sigma_n}{\{\sigma_n\}}{0}$,
  where $\sigma_i$ is the substitution in $\Sigma$ mapping $x$ to $t_i$ ($n$ is
  the number of times the weak quantifier was instantiated).
  \item $\Et{\forall x. A}{\Sigma}{1} = \Pi x. A +^{\alpha}
  \Et{A\sigma'}{\{\sigma'\}}{1}$ where $\sigma'$ is the substitution in
  $\Sigma$ mapping $x$ to $\alpha$.
  \item $\Et{\exists x. A}{\Sigma}{0} = \Pi x. A +^{\alpha}
  \Et{A\sigma'}{\{\sigma'\}}{0}$ where $\sigma'$ is the substitution in
  $\Sigma$ mapping $x$ to $\alpha$.
  \item $\Et{\exists x. A}{\Sigma}{1} = \Lambda x. A +^{t_1}
  \Et{A\sigma_1}{\{\sigma_1\}}{1} ... +^{t_n} \Et{A\sigma_n}{\{\sigma_n\}}{1}$,
  where $\sigma_i$ is the substitution in $\Sigma$ mapping $x$ to $t_i$ ($n$ is
  the number of times the weak quentifier was instatiated).
\end{itemize}

Note that the term $\alpha$ used for the strong quantifiers is determined by the
substitution set $\Sigma$. If the eigenvariable condition is not satisfied in
these substitutions, then the resulting expansion tree will not be a proof of
the formula.
\end{definition}

Using the $\Et{F}{\sigma}{p}$ transformation, it is also possible to define the
expansion sequent $\varepsilon$ from a sequent $S$.

\begin{definition}
\label{def:sequent_to_es}
Let $S: A_1, ..., A_n \vdash B_1, ..., B_m$ be a sequent with pairwise distinct
bound variables and $\sigma$ a set of substitutions for those variables such
that each strongly quantified variable is bound to exactly one term. Then we
define $\Ets{S}{\sigma}$ as the expansion sequent $\Et{A_1}{\sigma}{0}, ...,
\Et{A_n}{\sigma}{0} \vdash \Et{B_1}{\sigma}{1}, ..., \Et{B_m}{\sigma}{1}$.
\end{definition}

Definitions \ref{def:formula_to_et} and \ref{def:sequent_to_es} show how to
build an expansion sequent from a sequent and a set of substitutions. The
requirement of pairwise distinct variables can be easily satisfied by a variable
renaming. The second requirement, that each variable of a strong quantifier is
bound only once, might not be true for arbitrary proofs. Fortunately, it holds
for the proofs we are dealing with, either because the input problem contains no strong
quantifiers, or because the end-sequent is skolemized. On the second case, it is
possible to deduce unique Eigenvariables for each strong quantifier and obtain
the expansion tree of the un-skolemized formula.

\begin{lemma}
\label{lmm:shallow_of_et}
$\Sh{\Et{F}{\sigma}{p}}{p} = F$
\end{lemma}

\begin{proof}
Follows from the definition of $\Et{F}{\sigma}{p}$ and $\Sh{E}{p}$.
\end{proof}

\begin{theorem}
\label{thm:validity}
A sequent $S$ with substitutions $\sigma$, such that each strongly quantified
variable in $S$ is bound exactly once, is valid iff the expansion sequent
$\Ets{S}{\sigma}$ is an expansion proof.
\end{theorem}

\begin{proof}
By the soundness and completeness of expansion sequents \cite{miller87}, we know that
an expansion sequent $\varepsilon$ is an expansion proof iff its shallow sequent
is valid.
From Lemma \ref{lmm:shallow_of_et} we have that the shallow sequent of
$\Ets{S}{\sigma}$ is $S$. Therefore, $S$ is valid iff $\Ets{S}{\sigma}$ is an
expansion proof.
\end{proof}

This theorem provides a ``sanity-check'' for the expansion sequents extracted
from proof objects. If it is an expansion proof, we know that, at least, the
end-sequent with the given substitutions is a tautology. Note that this does not
provide a check for the proof, as it is not validating each inference applied,
but only if the claimed instantiations \emph{can} actually lead to a proof.
\vspace{-0.5cm}



\subsection{SMT proofs}
\label{sec:verit}
\vspace{-0.2cm}

SMT (\emph{Satisfiability Modulo Theory}) is a decision procedure for
first-order formulas with respect to a background theory. It can be seen as a
generalization of SAT problems.
VeriT\footnote{\url{http://www.verit-solver.org/}} is an
open-source SMT-solver which is complete for quantifier-free formulas with
uninterpreted functions and difference logic on reals and integers. For this
work we have used the proof objects produced by VeriT on the \texttt{QF\_UF}
(quantifier-free formulas with uninterpreted function symbols) problems of the
SMT-LIB\footnote{\url{http://smt-lib.org/}}. The background theory in this case was the equality
theory composed by the axioms (symmetry and reflexivity are implicit):
\vspace{-0.3cm}
\begin{align*}
&\forall x_0 ... \forall x_n. (x_0 = x_1 \wedge ... \wedge x_{n-1} = x_n \rightarrow x_0 = x_n)\\
&\forall x_0 ... \forall x_n \forall y_0 ... \forall y_n.
((x_0 = y_0 \wedge ... \wedge x_n = y_n \rightarrow f(x_0, ..., x_n) = f(y_0, ..., y_n))\\
&\forall x_0 ... \forall x_n \forall y_0 ... \forall y_n.
(x_0 = y_0 \wedge ... \wedge x_n = y_n \wedge p(x_0, ..., x_n) \rightarrow p(y_0, ..., y_n))\\
\end{align*}
\vspace{-1.2cm}

The proofs generated are composed of CNF transformations and a resolution
refutation, whose leaves are either one
of the quantifier-free formulas from the input problem or an instance of an
equality axiom. 
The proof object consists of a comprehensive list of labelled clauses used in the
resolution proof and their origin. They are either an input clause, without
ancestors, or the result of an inference rule on other clauses, which is
specified via the labels. VeriT's proof is purely propositional and no substitutions
are involved, since the axioms are quantifier-free and contain no free-variables.

The input problem is propositional, therefore the only substitutions needed were the ones
instantiating the (weak) quantifiers of the equality axioms\footnote{Observe
that we do not need any information from the inference steps.}. These are found by
collecting the ground instances of these axioms occurring on the leaves of the
resolution proof and using a first-order matching algorithm. By matching the
instances with the appropriate axiom (without the quantifiers), we can obtain the
substitutions for the quantified variables.
Given those substitutions and the quantified axioms, we can build the
expansion trees. 
It is worth noting that the quantified equality axioms (i.e.,
transitivity, symmetry, reflexivity, etc.) are build
internally in GAPT, since these are not part of the proof object. Also, the
reflexivity instances needed are computed separately, since these are implicit
in veriT.
The expansion
tree of the (propositional) input formula can be built with an empty set of
substitutions.
%
Since these are unsatisfiability proofs, all expansion trees will be on the left
side of the expansion sequent.
\vspace{-0.5cm}

\subsection{Connection proofs}
\label{sec:leancop}
\vspace{-0.2cm}

Connection calculi is a set of formalisms for deciding first-order classical
formulas which consists on connecting unifiable literals of opposite polarities
from the input. 
Proof search in these calculi is characterized as goal-oriented and, in general,
non-confluent. LeanCoP\footnote{\url{http://leancop.de/}} is a connection based theorem 
prover that implements a
series of techniques for reducing the search space and making proof search
feasible \cite{otten10}. 
Although its strategy is incomplete, it achieves very good performance
in practice. For this work, leanCoP 2.2 was used.
It can be
obtained from the CASC24 competition
website\footnote{\url{http://pages.cs.miami.edu/~tptp/CASC/24/Systems.tgz}} or,
alternatively, executed online at
SystemOnTPTP\footnote{\url{http://pages.cs.miami.edu/~tptp/cgi-bin/SystemOnTPTP}}.

Given an input problem (a set of axioms and conjectures in the language of
first-order logic), leanCoP will negate the axioms, skolemize the formulas
and translate them into a disjunctive normal form (DNF). It works with a positive
representation of the problem and uses a special DNF
transformation that is more suitable for connection proof search
\cite{otten10}. The prover also adds equality axioms when necessary. 
%
LeanCoP is able to produce proof objects in four different formats 
For this work, we have used
\texttt{leantptp}, which is closer to the TPTP (thousands of problems for
theorem provers) specification \cite{tptp}. 
The output
file is divided in three parts: (1) input formulas; (2) clauses generated from
the DNF transformation of the input and equality axioms; and (3) proof
description. Each part is described using a set of predicates with the relevant
information.

In part (1), the formulas from the input file are listed and named. Their
variables are renamed such that they are pairwise distinct. Moreover,
formulas are annotated with respect to their role, e.g, axiom or conjecture. Part (2)
contains the clauses, in the form of a list of literals, that resulted from the
disjunctive normal form transformation. This can either be the
regular naive  DNF translation or a \emph{definitional clausal form
transformation}, which assigns new predicates to some formulas. Each clause is
numbered and associated with the name of the formula that generated it. Equality
axioms are labelled with a special keyword, since they do not come from any
transformation on the input formulas. The proof \emph{per se} is in part (3),
where each line is an inference rule. It contains the number of the clause
to which the inference was applied, the bindings used (if any) and the
resulting clause.

For building the expansion trees of the input formulas we need the
substitutions used in the proof and the Skolem terms introduced during
Skolemization. The substitutions will be the terms of the expansion tree's weak
quantifiers and the Skolem terms, translated to variables, will be the expansion
tree's strong quantifier terms.
In the leanCoP proofs, Skolem terms have a specific syntax, so they can be
identified and parsed
as ``Eigenvariables''. We use this approach to get an expansion proof of the
original problem, instead of the skolemized problem. Since each strong
quantifier is replaced by exactly one Skolem term, the condition for the set of
substitutions in Definition \ref{def:formula_to_et} is satisfied.

The collection of terms used for the weak quantifiers is a bit more involved due
to variable renaming.
The quantified variables in the input formula are renamed during the clausal
normal form transformation. This means that the sets of variables occurring in
the original problem and in the clauses are disjoint. The substitutions used in
the proof are given with respect to the clauses' variables, but we are
interested in building expansion trees of the input formulas. We need therefore
to find a way to map the variables in the clauses to the variables in the input
formulas.

The solution found was to implement in GAPT the definitional clausal
form transformation, trying to remain as faithful as possible to the one
leanCoP uses, but without the variable renaming. After applying our
transformation to the input formulas, we try to match the clauses obtained
to the clauses from the proof object. The first-order matching algorithm returns
a substitution if a match is found. Such substitution maps strongly quantified
variables to ``Eigenvariables'' (the result of parsing Skolem terms), and weakly
quantified variables to their renamed versions used in the clauses. 
By composing this substitution with
the ones obtained from the bindings in the proof, we are able to correctly
identify the terms used for each quantified variable in the input formulas.
\vspace{-0.5cm}

\section{Results}
\label{sec:results}
\vspace{-0.2cm}

We were able to import as expansion trees all the 142 proof objects provided to
us by the veriT team, and all but one under one minute.
The expansion sequents generated have been used as
input for the cut-introduction algorithm \cite{hetzl14} and some of their
features (e.g. high number of instances) have motivated improvements to
the algorithm.
As for leanCoP, our database consists of 3043 proofs of problems from the TPTP
library \cite{tptp}. Of those, we can successfully import 1224 as expansion sequents. Some
errors still occur while parsing and matching (e.g. our generated clauses do not
have the same literal ordering as the clauses in the proof file),
but we are working to
increase the success rate.

Getting proofs from various theorem provers in the shape of expansion sequents
allows us to do a number of interesting things.
First of all, one can visualize the end-sequent and the
instances used of each quantified formula. This is much more comfortable and
easier to grasp than a raw text file. 
%
%
It is also possible to check whether the instances used lead indeed to a proof
of the end-sequent. This is reduced to checking if the deep sequent of the
expansion sequent is a tautology (which can be done, as this sequent is
propositional) and if the dependency relation is acyclic. In case the expansion
sequent is a proof, we can build an LK proof from it, using the dependency
relation to decide the order in which quantifiers are introduced \cite{miller87}.
Finally, one can attempt proof compression and discovery of lemmas using the
cut-introduction algorithm \cite{hetzl14}.

All of these functionalities are implemented in GAPT. The system comes with an
interactive command line where commands for loading proofs, opening
\texttt{prooftool}, introducing cuts, eliminating cuts, building an LK proof
from an expansion sequent, among others, can be issued. Some examples of proofs
imported and their visualizations can be found at
\url{https://www.logic.at/staff/giselle/examples.pdf}.

\vspace{-0.5cm}

\section{Related Work}
\label{sec:related}
\vspace{-0.2cm}

Other projects and tools also address the issues of proof visualization and
checking. For proofs in the TPTP language in particular, there is IDV
\cite{Sutcliffe2007}, which provides an interactive interface for manipulating
the DAG representing a derivation. This tool focuses solely on visualization
of proofs in the TPTP format. Our work aims on a more general framework, of
which visualization is only a small part. 
We are also capable to import different proof objects, not only those in the
TPTP language.

As for proof checking, \cite{KaliszykCPP2015} proposes a check of leanCoP proofs
in HOL Light while \cite{Armand2011} shows how to check SAT and SMT proofs using
Coq. The former paper involved re-implementing leanCoP's kernel in HOL
Light, which differs a lot from our approach of simply parsing the outputs of
theorem provers. In the latter, proofs produces by SAT/SMT theorem provers
are certified by Coq. We must clarify that, given the information needed to
produce expansion proofs, it is not fair to claim we are checking proof objects,
but we merely have a sanity check that the instances used by the theorem prover
actually lead to a proof of the proposed theorem. Such compromise makes sense if
we want a framework general enough to deal with different proof objects, without
asking any change on the side of theorem provers.

Finally, it is worth mentioning ProofCert \cite{proofcert}, a research project
with the aim of developing a theoretical framework for proof representation. In
order not to make such compromise, and actually check each step of each proof
for various different proof objects, a solid foundation of proof specification
needs to be developed. While this does not happen, this work shows how it is
still possible to combine existing proof objects into one representation.
\vspace{-0.5cm}

\section{Conclusion}
\label{sec:conclusion}
\vspace{-0.2cm}

We have shown how SMT and Connection proofs can be both imported as expansion
sequents. The information needed from the proof objects is just the end-sequent
being proven and a set of instances used for the quantified formulas. For both
cases presented we relied on a first-order matching algorithm, but this
requirement can be lifted if all substitutions are provided directly in the
proof object.

The representation using expansion sequents serves various purposes. It provides
an easy proof visualization, a simple checking procedure, LK proof construction
and introduction of cuts.

This is an ongoing work, and we hope to have many developments in the near future. In
particular, the difficulties in importing leanCoP proofs remain to be resolved. This
procedure also offers a lot of room for optimization. Once we have a big enough
set of parsed leanCoP proofs, we will add those to the benchmark used in the
cut-introduction algorithm. As for veriT proofs, we
plan to test bigger examples, as the ones provided are only a small
subset from the SMT-LIB.

Another future goal is importing other formats from other provers and comparing
the different proofs for the same input problem. We also aim on integrating 
a check for whether the obtained expansion sequent is an expansion proof in the
import function.
\vspace{-0.5cm}
%
%
%
%
\bibliographystyle{eptcs}
\bibliography{references}

\begin{thebibliography}{10}
\providecommand{\bibitemdeclare}[2]{}
\providecommand{\surnamestart}{}
\providecommand{\surnameend}{}
\providecommand{\urlprefix}{Available at }
\providecommand{\url}[1]{\texttt{#1}}
\providecommand{\href}[2]{\texttt{#2}}
\providecommand{\urlalt}[2]{\href{#1}{#2}}
\providecommand{\doi}[1]{doi:\urlalt{http://dx.doi.org/#1}{#1}}
\providecommand{\bibinfo}[2]{#2}

\bibitemdeclare{incollection}{Armand2011}
\bibitem{Armand2011}
\bibinfo{author}{Michael \surnamestart Armand\surnameend},
  \bibinfo{author}{Germain \surnamestart Faure\surnameend},
  \bibinfo{author}{Benjamin \surnamestart Grégoire\surnameend},
  \bibinfo{author}{Chantal \surnamestart Keller\surnameend},
  \bibinfo{author}{Laurent \surnamestart Théry\surnameend} \&
  \bibinfo{author}{Benjamin \surnamestart Werner\surnameend}
  (\bibinfo{year}{2011}): \emph{\bibinfo{title}{A Modular Integration of
  SAT/SMT Solvers to Coq through Proof Witnesses}}.
\newblock In: {\sl \bibinfo{booktitle}{CPP}}, \bibinfo{series}{Lecture Notes in
  Computer Science}, \bibinfo{publisher}{Springer Berlin Heidelberg}, pp.
  \bibinfo{pages}{135--150}, \doi{10.1007/978-3-642-25379-9\_12}.

\bibitemdeclare{article}{ceres}
\bibitem{ceres}
\bibinfo{author}{Matthias \surnamestart Baaz\surnameend} \&
  \bibinfo{author}{Alexander \surnamestart Leitsch\surnameend}
  (\bibinfo{year}{2000}): \emph{\bibinfo{title}{Cut-elimination and
  Redundancy-elimination by Resolution}}.
\newblock {\sl \bibinfo{journal}{Journal of Symbolic Computation}}
  \bibinfo{volume}{29}(\bibinfo{number}{2}), pp. \bibinfo{pages}{149--176},
  \doi{10.1006/jsco.1999.0359}.

\bibitemdeclare{inproceedings}{dunchev13}
\bibitem{dunchev13}
\bibinfo{author}{Cvetan \surnamestart Dunchev\surnameend},
  \bibinfo{author}{Alexander \surnamestart Leitsch\surnameend},
  \bibinfo{author}{Tomer \surnamestart Libal\surnameend},
  \bibinfo{author}{Martin \surnamestart Riener\surnameend},
  \bibinfo{author}{Mikheil \surnamestart Rukhaia\surnameend},
  \bibinfo{author}{Daniel \surnamestart Weller\surnameend} \&
  \bibinfo{author}{Bruno~Woltzenlogel \surnamestart Paleo\surnameend}
  (\bibinfo{year}{2013}): \emph{\bibinfo{title}{{PROOFTOOL:} a {GUI} for the
  {GAPT} Framework}}.
\newblock In: {\sl \bibinfo{booktitle}{10th {UITP}}}, {\sl
  \bibinfo{series}{{EPTCS}}} \bibinfo{volume}{118}, pp. \bibinfo{pages}{1--14},
  \doi{10.4204/EPTCS.118.1}.

\bibitemdeclare{article}{weller13}
\bibitem{weller13}
\bibinfo{author}{Cvetan \surnamestart Dunchev\surnameend},
  \bibinfo{author}{Alexander \surnamestart Leitsch\surnameend},
  \bibinfo{author}{Mikheil \surnamestart Rukhaia\surnameend} \&
  \bibinfo{author}{Daniel \surnamestart Weller\surnameend}
  (\bibinfo{year}{2013}): \emph{\bibinfo{title}{{CERES} for First-Order
  Schemata}}.
\newblock {\sl \bibinfo{journal}{CoRR}} \bibinfo{volume}{abs/1303.4257},
  \doi{10.1007/978-3-662-46906-4\_8}.

\bibitemdeclare{article}{gentzen35}
\bibitem{gentzen35}
\bibinfo{author}{Gerhard \surnamestart Gentzen\surnameend}
  (\bibinfo{year}{1935}): \emph{\bibinfo{title}{{U}ntersuchungen \"uber das
  logische {S}chlie{\ss}en {I}}}.
\newblock {\sl \bibinfo{journal}{Mathematische Zeitschrift}}
  \bibinfo{volume}{39}(\bibinfo{number}{1}), pp. \bibinfo{pages}{176--210},
  \doi{10.1007/BF01201353}.

\bibitemdeclare{inproceedings}{hetzl14}
\bibitem{hetzl14}
\bibinfo{author}{Stefan \surnamestart Hetzl\surnameend},
  \bibinfo{author}{Alexander \surnamestart Leitsch\surnameend},
  \bibinfo{author}{Giselle \surnamestart Reis\surnameend},
  \bibinfo{author}{Janos \surnamestart Tapolczai\surnameend} \&
  \bibinfo{author}{Daniel \surnamestart Weller\surnameend}
  (\bibinfo{year}{2014}): \emph{\bibinfo{title}{Introducing Quantified Cuts in
  Logic with Equality}}.
\newblock In: {\sl \bibinfo{booktitle}{7th {IJCAR}}}, {\sl
  \bibinfo{series}{Lecture Notes in Computer Science}} \bibinfo{volume}{8562},
  \bibinfo{publisher}{Springer}, pp. \bibinfo{pages}{240--254},
  \doi{10.1007/978-3-319-08587-6\_17}.

\bibitemdeclare{inproceedings}{KaliszykCPP2015}
\bibitem{KaliszykCPP2015}
\bibinfo{author}{Cezary \surnamestart Kaliszyk\surnameend},
  \bibinfo{author}{Josef \surnamestart Urban\surnameend} \&
  \bibinfo{author}{Ji\v{r}i \surnamestart Vysko\v{c}il\surnameend}
  (\bibinfo{year}{2015}): \emph{\bibinfo{title}{Certified Connection Tableaux
  Proofs for HOL Light and TPTP}}.
\newblock \bibinfo{series}{CPP '15}, \bibinfo{publisher}{ACM},
  \bibinfo{address}{New York, NY, USA}, pp. \bibinfo{pages}{59--66},
  \doi{10.1145/2676724.2693176}.

\bibitemdeclare{article}{miller87}
\bibitem{miller87}
\bibinfo{author}{Dale \surnamestart Miller\surnameend} (\bibinfo{year}{1987}):
  \emph{\bibinfo{title}{A compact representation of proofs}}.
\newblock {\sl \bibinfo{journal}{Studia Logica}}
  \bibinfo{volume}{46}(\bibinfo{number}{4}), pp. \bibinfo{pages}{347--370},
  \doi{10.1007/BF00370646}.

\bibitemdeclare{unpublished}{proofcert}
\bibitem{proofcert}
\bibinfo{author}{Dale \surnamestart Miller\surnameend} (\bibinfo{year}{2011}):
  \emph{\bibinfo{title}{ProofCert: Broad Spectrum Proof Certificates}}.
\newblock \bibinfo{note}{ERC Advanced Grant 2012-2016}.

\bibitemdeclare{article}{otten10}
\bibitem{otten10}
\bibinfo{author}{Jens \surnamestart Otten\surnameend} (\bibinfo{year}{2010}):
  \emph{\bibinfo{title}{Restricting backtracking in connection calculi}}.
\newblock {\sl \bibinfo{journal}{{AI} Commun.}}
  \bibinfo{volume}{23}(\bibinfo{number}{2-3}), pp. \bibinfo{pages}{159--182},
  \doi{10.3233/AIC-2010-0464}.

\bibitemdeclare{article}{robinson65}
\bibitem{robinson65}
\bibinfo{author}{J.~A. \surnamestart Robinson\surnameend}
  (\bibinfo{year}{1965}): \emph{\bibinfo{title}{A Machine-Oriented Logic Based
  on the Resolution Principle}}.
\newblock {\sl \bibinfo{journal}{J. ACM}}
  \bibinfo{volume}{12}(\bibinfo{number}{1}), pp. \bibinfo{pages}{23--41},
  \doi{10.1145/321250.321253}.

\bibitemdeclare{article}{tptp}
\bibitem{tptp}
\bibinfo{author}{G.~\surnamestart Sutcliffe\surnameend} (\bibinfo{year}{2009}):
  \emph{\bibinfo{title}{{The TPTP Problem Library and Associated
  Infrastructure: The FOF and CNF Parts, v3.5.0}}}.
\newblock {\sl \bibinfo{journal}{Journal of Automated Reasoning}}
  \bibinfo{volume}{43}(\bibinfo{number}{4}), pp. \bibinfo{pages}{337--362},
  \doi{10.1007/s10817-009-9143-8}.

\bibitemdeclare{article}{Sutcliffe2007}
\bibitem{Sutcliffe2007}
\bibinfo{author}{Steven \surnamestart Trac\surnameend}, \bibinfo{author}{Yury
  \surnamestart Puzis\surnameend} \& \bibinfo{author}{Geoff \surnamestart
  Sutcliffe\surnameend} (\bibinfo{year}{2007}): \emph{\bibinfo{title}{An
  Interactive Derivation Viewer}}.
\newblock {\sl \bibinfo{journal}{Electronic Notes in Theoretical Computer
  Science}} \bibinfo{volume}{174}(\bibinfo{number}{2}), pp. \bibinfo{pages}{109
  -- 123}, \doi{10.1016/j.entcs.2006.09.025}.
\newblock \bibinfo{note}{Proceedings of the 7th Workshop on User Interfaces for
  Theorem Provers (UITP 2006)}.

\end{thebibliography}
\end{document}